\documentclass[12pt]{article}
\pdfoutput=1

\usepackage[utf8]{inputenc}
\usepackage[left=2.55cm, right=2.55cm, top=2.55cm, bottom=2.55cm]{geometry}
\usepackage{amsmath,amssymb,amsbsy,amsthm}
\usepackage{slashed}
\usepackage{xcolor}
\usepackage{graphicx}
\usepackage{url}
\usepackage{cancel}
\usepackage{cite}
\usepackage[colorlinks=true,allcolors=darkpurple,pdfborder={0 0 0},linktocpage=false]{hyperref}
\usepackage{tabularx,booktabs}
\usepackage{multicol}
\usepackage{units}
\usepackage{xspace}
\usepackage[labelfont=bf]{caption}
\usepackage[section]{placeins}
\usepackage{enumitem}
\usepackage{ulem}
\usepackage{subcaption}
\usepackage{bm}

\definecolor{darkred}{rgb}{0.6,0,0}
\definecolor{darkpurple}{rgb}{0.5,0,0.5}

\def\L{\mathcal{L}}

\def\U{\mathcal{U}}
\def\hc{\text{h.c.}}

\def\id{\mathbb{I}}
\def\rank{\text{rank}}
\def\npar{\texttt{\#}}

\def\U1L{$\mathrm{U(1)}_L$}

\newtheorem{theorem}{Theorem}
\newtheorem{lemma}{Lemma}

\definecolor{aqua}{rgb}{0.4, 0.6, 0.7}
\definecolor{avblue}{rgb}{0.0, 0.0, 0.8}
\definecolor{asparagus}{rgb}{0.53, 0.66, 0.42}

\newcommand{\AddrIFIC}{%
  Instituto de F\'{i}sica Corpuscular, CSIC-Universitat de Val\`{e}ncia, 46980 Paterna, Spain}

\newcommand{\AddrFISTEO}{%
  Departament de F\'{\i}sica Te\`{o}rica, Universitat de Val\`{e}ncia, 46100 Burjassot, Spain}

\begin{document}

\begin{center}
\vspace*{15mm}

\vspace{1cm}
{\Large \bf 
Minimal Majorana neutrino mass models
} \\
\vspace{1cm}

{\bf Antonio Herrero-Brocal$^{\text{a}}$, Avelino Vicente$^{\text{a,b}}$}

 \vspace*{.5cm} 
 $^{(\text{a})}$ \AddrIFIC \\\vspace*{.2cm} 
 $^{(\text{b})}$ \AddrFISTEO

 \vspace*{.3cm}
\href{mailto:antonio.herrero@ific.uv.es}{antonio.herrero@ific.uv.es},
\href{mailto:avelino.vicente@ific.uv.es}{avelino.vicente@ific.uv.es}
\end{center}

\vspace*{10mm}
\begin{abstract}\noindent\normalsize
We present an extension of the Casas-Ibarra parametrization that applies to all possible Majorana neutrino mass models. This framework allows us to systematically identify minimal models, defined as those with the smallest number of free parameters. We further analyze the phenomenologically relevant combination of the Yukawa matrix, $y^\dagger y$, and show that in certain scenarios it exhibits an unexpected reduction in the number of free parameters, depending on just one real degree of freedom. Finally, the application of our results is illustrated in specific models, which can be tested or falsified due to their definite experimental predictions in heavy neutrino and charged lepton flavor violating decays.
\end{abstract}

\section{Introduction}\label{sec:Intro}

The lack of neutrino masses in the Standard Model (SM) of particle physics is one of its most important shortcomings. Nowadays, after decades of experimental and theoretical progress in the determination of neutrino oscillation parameters~\cite{deSalas:2020pgw}, an extension of the SM that can account for neutrino masses and lepton mixing, and explain its particularities, is more necessary than ever.

Many neutrino mass models have been put forward along the years (see~\cite{Boucenna:2014zba,King:2017guk,Cai:2017jrq} for some recent reviews). Even within the most popular option, Majorana neutrinos, there is a broad variety of models, which differ by their symmetries and particle contents, the energy scale at which neutrinos acquire their masses or the level in perturbation theory at which they are induced, among other features. Testing these models experimentally is of upmost importance. In principle, there are many ways to do it, ranging from searches of new states at high-energy colliders to the study of rare processes in low-energy experiments. However, in most models, the large number of free parameters makes it impossible to make any definite predictions. In these cases, fully falsifying a model becomes very hard, and one can only search for some generic signals and disfavor it in case they are not observed.

In this work we consider minimal Majorana neutrino mass models, defined by their small number of free parameters. In contrast to the most common scenarios, these models allow for definite numerical predictions for some observables which, if measured, would either provide a strong support to these models or rule them out completely. In order to characterize them, we will introduce an extended Casas-Ibarra parametrization that will allow us to easily count the number of free parameters in a given model. Armed with this tool, we will prove the existence of models with a very minimal number of free parameters, just two real degrees of freedom. These scenarios are absolutely predictive, with some observables, such as the rates of flavor violating decays $\ell_\alpha \to \ell_\beta \gamma$, with $\alpha \neq \beta$, depending on a single real number.

The rest of the paper is structured as follows. In the next Section we present the extended Casas-Ibarra parametrization. In Section~\ref{sec:CIPheno} we consider some combinations of Yukawa matrices, of special relevance for the phenomenology of many models, and find a simple formula for the number of free parameters. We also point out some cases of interest that depend on just one free real parameter. Some phenomenological applications of our general analysis to specific example models are given in Sec.~\ref{sec:appli}. Finally, we summarize and conclude in Sec.~\ref{sec:summ}. Additional details, including some mathematical proofs, are given in the Appendices.

\section{The extended Casas-Ibarra parametrization}
\label{sec:Casas}

The neutrino mass matrix of any model that induces Majorana neutrino masses can be written as
\begin{equation} \label{eq:mnu}
m_\nu = f \, Y^T \, M \, Y \, ,
\end{equation}
where $f$ is a complex number, $Y$ is a general $n \times 3$ matrix and $M$ is a symmetric $n\times n$ matrix.~\footnote{It is easy to convince oneself that this is the case. It is straightforward for models in which neutrinos couple to the neutrino mass mediators by means of just one Yukawa matrix, $Y$. In this case, the requirement $m_\nu = m_\nu^T$ leads directly to Eq.~\eqref{eq:mnu}. In models involving two distinct Yukawa matrices, $y_1$ and $y_2$, one can always rewrite the neutrino mass matrix in the form of Eq.~\eqref{eq:mnu}, as shown in Appendix~\ref{app:Master}.} The Yukawa matrix $Y$ plays a central role in phenomenological applications and is typically expressed in terms of the model parameters contained in $f$ and $M$, along with the entries of the $m_\nu$ matrix, measured in neutrino oscillation experiments~\cite{Casas:2001sr,Cordero-Carrion:2018xre,Cordero-Carrion:2019qtu,Heeck:2012fw,Chen:2025cor}. This is commonly referred to as a \textit{parametrization} of the Yukawa matrix $Y$, which simply corresponds to a solution of Eq.~\eqref{eq:mnu}. Since the solutions to this matrix equation are not unique, all parametrizations lead to some free parameters in $Y$, degrees of freedom that cancel out in the combination $Y^T \, M \, Y$. The numerical values of these parameters can be taken freely (sometimes within certain ranges), without affecting Eq.~\eqref{eq:mnu}. A general and complete parametrization applicable to any Majorana neutrino mass model is provided by the so-called master parametrization~\cite{Cordero-Carrion:2018xre,Cordero-Carrion:2019qtu}. However, here we are interested in a parametrization that allows us to count the free parameters in $Y$ in a more straightforward way. Since Eq.~\eqref{eq:mnu} is formally equivalent to the Type-I Seesaw formula~\cite{Minkowski:1977sc,Yanagida:1979as,Mohapatra:1979ia,GellMann:1980vs,Schechter:1980gr}, this can be achieved by means of a simple extension of the Casas-Ibarra parametrization~\cite{Casas:2001sr}.

Since both $M$ and $m_\nu$ are symmetric matrices, they admit a Takagi decomposition, 
\begin{align} \label{eq:Udef}
U^T m_\nu \, U = D_m \equiv \textup{diag}(m_1,m_2,m_3) \,, & &  V^* M V^\dagger = D_M \equiv \textup{diag}(M_1,M_2,\dots, M_N) \, ,
\end{align}
where $U$ and $V$ are unitary matrices. We can make use of this to write
\begin{align} 
 D_{m}  &= f \, U^T \, Y^T \,V^T \,D_M \,V \,Y \, U\, .
\end{align}
The symmetry of $D_m$ allows us to define its square root, $D_{\sqrt{m}}$, such that $D_{\sqrt{m}}^T D_{\sqrt{m}} = D_m$. We define $D_{\sqrt{m}}$, depending on the rank of $m_\nu$, denoted $r_m$, as the $r_m \times 3$ matrix 
\begin{equation} \label{eq:Dsqrtm}
D_{\sqrt{m}} = \left\{ \begin{array}{cl}
\textup{diag}\left( \sqrt{m_1},\sqrt{m_2},\sqrt{m_3} \right) \, , \quad & \textup{if} \, \,\, r_m =3 \, , \\
& \\
\begin{pmatrix}
  0 &\sqrt{m_2} & 0\\
  0 & 0        & \sqrt{m_3}\\
 \end{pmatrix} \cdot P \, , \quad & \textup{if} \, \, \, r_m = 2 \, ,
\end{array} \right.
\end{equation}
where
\begin{equation} \label{eq:Pdef}
P = \left\{ \begin{array}{cl}
\id_3 \, , \quad & \text{for neutrino normal hierarchy (NH)} \, , \\
& \\
\begin{pmatrix}
0 & 0 & 1\\
0 & 1 & 0\\
1 & 0 & 0
\end{pmatrix} \, , \quad & \text{for neutrino inverted hierarchy (IH)} \, ,
\end{array} \right.
\end{equation}
with $\id_n$ the $n \times n$ identity matrix. Defined in this form, $D_{\sqrt{m}}$ is always a full rank matrix. This implies that, even in the case $r_m=2$, the matrix $D_{\sqrt{m}}$ admits a right inverse, $D_{\sqrt{m}}^{-1}$.~\footnote{This is relevant for the discussion that follows, since in this case one can claim that the two matrices are related by $\mathcal{R}$. The argument is as follows: let $A$ be an $m \times n$  matrix and $B$ an $s \times n$ matrix such that $A^T A = B^T B$. If $m < n$ and $A$ is full rank, then $A$ admits a right inverse, $A^{-1}$. Thus, the equation
\begin{equation}
 B = \mathcal{R} \, A \, ,
\end{equation}
has a unique solution for $\mathcal{R}$ and this solution satisifies
\begin{equation}
 \mathcal{R}^T \, \mathcal{R} = (A^{-1})^T B^T B  A^{-1}  =(A^{-1})^T A^T A  A^{-1} = \id_{m} \, .
\end{equation}
} Therefore, we obtain the relation
\begin{equation} \label{eq:YRrelation}
 \sqrt{ f}\,\sqrt{D_M}\,  V \, Y \, U \,  = \mathcal{R} \, D_{\sqrt{m}}\, ,
\end{equation}
where $\mathcal{R}$ is an $n \times r_m$ matrix satisfying $\mathcal{R}^T \mathcal{R} = \id_{r_m}$ and $\sqrt{D_M}$ is defined by $\sqrt{D_M}_{ii} =\sqrt{(D_M)_{ii}}$. To obtain an analytical expression for $Y$, we decompose $D_M$ into two blocks: one containing the real and positive singular values and the other consisting of zeros,
\begin{equation}
\sqrt{D_M}\equiv 
\begin{pmatrix}
\sqrt{D_{r_M}} & 0\\
0                     & 0
\end{pmatrix} \, ,
\end{equation}
where $D_{r_M}$ is an $r_M \times r_M$ matrix with $r_M \equiv \text{rank}(M)$. In the same spirit, we split $\mathcal{R}$ into two blocks,
\begin{equation}
\mathcal{R} = 
\begin{pmatrix}
R_{r_M}\\
R_{n-r_M}
\end{pmatrix} \, .
\end{equation}
By applying Eq.~\eqref{eq:YRrelation} it is trivial to show that $R_{n-r_M} = 0_{(n-r_M)\times r_m}$ and $R_{r_M}^T R_{r_M} \equiv R^T R = \id_{r_m}$. Thus, the parametrization for $Y$ is given by
\begin{equation} \label{eq:Yparam}
Y \equiv
\begin{pmatrix}
Y_{r_M} \\
Y_{n-r_M}
\end{pmatrix} = \frac{1}{\sqrt{f}} \,
V^\dagger  \, \begin{pmatrix}
 \, \sqrt{D_{r_M}}^{-1} R   \, D_{\sqrt{m}} \, U^\dagger \,\\
X
\end{pmatrix} \, ,
\end{equation}
where $Y_{r_M}$ is a general $r_M \times 3$ matrix while $X$ is a completely free $(n-r_M)\times 3$ matrix. Eq.~\eqref{eq:Yparam} constitutes the extended Casas-Ibarra parametrization.
%%%%%%%%%%%%%%%%%%%%%%%%%%%%%%%%%%%%%%%%%
%%%%%%%%%%%%%%%%%%%%%%%%%%%%%%%%%%%%%%%%%

\section{Yukawa combinations of phenomenological relevance}
\label{sec:CIPheno}

The extended Casas-Ibarra parametrization presented in the previous Section allows one to compute the number of free parameters in the Yukawa matrices of a given Majorana neutrino mass model, $\npar_Y$ (in the following, $\npar_A$ denotes the number of \textit{real} parameters in the matrix $A$). This number is of crucial importance when comparing experimental data with the predictions of the model. In particular, the parameters of the neutrino sector are expected to play a central role in a wide range of phenomena, including lepton flavor violation, leptogenesis, and collider signatures, among others. One might naively expect that all the degrees of freedom in $Y$ affect each observable independently. If this were the case, at least $\npar_Y$ distinct measurements would be required to fully test a model characterized by a given set of observables. In practice, however, degeneracies are common. Most observables are typically governed by combinations of the form $y_i^\dagger y_i$, where $y_i$ is a submatrix of $Y$ (including, in some cases, $Y$ itself). As a consequence of this, the number of parameters relevant for a specific observable can be smaller than $\npar_Y$. In this section, with the goal of finding minimal models, we focus on counting the free parameters of two matrices particularly relevant for phenomenology: $Y$ and $y_i^\dagger y_i$.
%%%%%%%%%%%%%%%%%%%%%%%%%%%%%%%%%%%%%%%%%%%%
%%%%%%%%%%%%%%%%%%%%%%%%%%%%%%%%%%%%%%%%%%%%
\subsection{\boldmath $Y$ matrix}
\label{subsec:freeparam}

This case follows trivially from Eq.~\eqref{eq:Yparam}. If the model does not impose additional constraints on the Yukawa couplings of the theory, then $X$ is a free $(n-r_M) \times 3$ matrix, which implies $\npar_X = 6(n-r_M)$. Moreover, the free parameters of $Y_{r_M}$ are encoded in $R$. To count $\npar_R$ we simply make use of the orthonormality of its columns. Therefore, the independent parameters in $Y$ are given by
\begin{equation} \label{eq:Pdef}
\npar_Y = 6n - 2 r_M \left(3-r_m \right)-r_m \left(r_m+1 \right) - \npar_{\text{extra}} 
\end{equation}
where $\npar_{\text{extra}}$ accounts for possible extra constraints on the Yukawa $Y$ (for instance, requiring it to be antisymmetric). One can read from this expression the minimal number of free parameters that a non-restricted model can have, i.e., a model without additional constraints on the Yukawa couplings. First, let us note that Eq.~\eqref{eq:mnu} implies
\begin{equation}
r_m \leq \text{min}\left( \rank\left(Y \right), \,\rank \left(M \right) \right) \, ,
\end{equation}
with $\rank\left(Y \right) \leq \text{min}(n,3)$ and $\rank \left(M\right)=r_M \leq n$. Thus, the minimal scenenario with $\npar_{\text{extra}}=0$ will be given by 
\begin{align} \label{eq:minimalcond}
  n=r_M =r_m = 2 \, ,
\end{align}
with
\begin{align}
  \npar_Y=2 \, .
\end{align}
We conclude that minimal Majorana neutrino mass models have just \textbf{\boldmath two real degrees of freedom in the Yukawa matrix $Y$}.

%%%%%%%%%%%%%%%%%%%%%%%%%%%%%%%%%%%%%%%%%%%%%%%
%%%%%%%%%%%%%%%%%%%%%%%%%%%%%%%%%%%%%%%%%%%%%%%
\subsection{\boldmath $y_i^\dagger y_i$ matrix}

We begin by providing a proper definition of $y_i$. In general the Yukawa contributing to neutrino masses, $Y_{r_m}$, is composed of several individual Yukawas. We define $y_i$ as the different $s_i\times 3$ submatrices of $Y_{r_M}$, with $s_i$ the dimension of each individual Yukawa:
\begin{equation}
Y = Y_{r_M} = \begin{pmatrix}
y_1\\
y_2\\
\vdots\\
y_p
\end{pmatrix} \, .
\end{equation}
Now, motivated by the search for minimal models, we assume that the matrix $X$ does not exist, i.e., $n = r_M$. Actually, this is already justified since minimal models verify $n=r_M$. Nevertheless, a similar procedure applies to the case with a non-zero $X$ matrix. Then, using Eq.~\eqref{eq:Yparam}, one can write
\begin{equation}
y_i = \frac{1}{\sqrt{f}} \, v_i^\dagger \sqrt{D_{r_M}}^{-1} \, R \, D_{\sqrt{m}} \, U^\dagger \, ,
\end{equation}
where $V$ is expressed in terms of the $r_M \times s_i$ submatrices $v_i$ as
\begin{equation}
V = \begin{pmatrix}
v_1 & v_2 & \dots & v_p
\end{pmatrix} \, .
\end{equation}
Furthermore, in this framework $y_i^\dagger y_i$ takes the form
\begin{equation}
y_i^\dagger y_i =\frac{1}{\left|f\right|}\, U \, D_{\sqrt{m}}^T \, R^\dagger \sqrt{D_{r_M}}^{-1} \, v_i \,v_i^\dagger \,  \sqrt{D_{r_M}}^{-1} R D_{\sqrt{m}}  U^\dagger \, .
\end{equation}
The unitarity of $V$ implies $v_i^\dagger v_j = \id_{s_i} \,\delta_{ij}$, but nothing can be said in general of the combination $v_i v_j^\dagger$. Then, one generally expects $\npar_{y_i^\dagger y_i} = \text{min} \left(\npar_{y_i}, 9\right) = \text{min} \left(\npar_{Y}, 9\right)$, with 9 the maximum for a Hermitian $3\times 3$ matrix. However, there are special cases in which this number is reduced. We shall call them \textit{reduced scenarios} and discuss in the following two examples.

%%%%%%%%%%%%%%%%%%%%%%%%%%%%%%%
%%%%%%%%%%%%%%%%%%%%%%%%%%%%%%%
\subsubsection{Reduced scenario 1}

This scenario is defined by the following conditions:
\begin{enumerate}
  \item $D_{r_M}$ is completely degenerate: $D_{r_M} = \Lambda \, \id_{r_M }$.
  \item $R$ is orthogonal.
  \item $v_i v_i^\dagger = \id_{r_M}$.
\end{enumerate}
Let us highlight two relevant facts about these conditions. First, the second condition requires $R$ to be square. Since $R$ is an $r_M \times r_m$ matrix with $r_m \leq 3$, we must have either $r_M = r_m = 3$ or $r_M = r_m = 2$. Secondly, the third condition is always satisfied when considering the full combination $Y^\dagger Y$, as in this case we are dealing with the full unitary matrix $V$.

Under these requirements, the matrix of interest can be written as
\begin{equation} \label{eq:paramH}
y_i^\dagger y_i = \frac{1}{\left|f\right|\Lambda} U \, D_{\sqrt{m}}^T \, R^\dagger R D_{\sqrt{m}}  U^\dagger \equiv \frac{1}{\left|f\right|\Lambda} U \, D_{\sqrt{m}}^T \, H D_{\sqrt{m}}  U^\dagger \, .
\end{equation}
Thus, we can parametrize it in terms of the Hermitian positive-definite orthogonal $r_m \times r_m$ matrix $H=R^\dagger R$. We refer to Appendix~\ref{app:HH0} for a rigorous study of this matrix and report here our main results. First of all, the number of real degrees of freedom contained in $H$ (and then in $y_i^\dagger y_i$) is
\begin{equation}
  \npar_H = \frac{r_m(r_m-1)}{2} \, .
\end{equation}
Explicit expressions for the cases $r_m=3$ and $r_m=2$ are given as follows:

\subsubsection*{$\boldsymbol{r_M=r_m=3}$}

In this case $H$ has three free parameters, $a,b,c \in \mathbb{R}$, and can be parametrized as
\begin{equation} \label{eq:rs1N3}
H= \frac{1}{x^2}
\begin{pmatrix}
c^2 + \left(a^2 + b^2\right)\sqrt{1 + x^2} & 
a b \left(-1 + \sqrt{1 + x^2}\right)+i a x^2  & 
a c\left(1-\sqrt{1 + x^2} \right)+ i b x^2  \\
a b  \left(-1 + \sqrt{1 + x^2}\right)-i a x^2 & 
b^2 + \left(a^2 + c^2\right)\sqrt{1 + x^2} & 
a b \left(-1 + \sqrt{1 + x^2}\right)+i c x^2  \\
a c\left(1-\sqrt{1 + x^2} \right) - i b x^2  & 
a b \left(-1 + \sqrt{1 + x^2}\right)- i c x^2 & 
a^2 + \left(b^2 + c^2\right)\sqrt{1 + x^2}
\end{pmatrix} \, ,
\end{equation}
with $x = \sqrt{a^2 +b^2 +c^2}$.

\subsubsection*{$\boldsymbol{r_M = r_m =2}$}

In this case $H$ can be written in terms of only one free parameter, $a \in \mathbb{R}$, with
\begin{equation} \label{eq:rs1}
H = 
\begin{pmatrix}
\sqrt{1+a^2} & i a\\
-i a         & \sqrt{1+a^2}
\end{pmatrix} \, .
\end{equation}

%%%%%%%%%%%%%%%%%%%
\subsubsection{Reduced scenario 2}

This scenario is defined by the following conditions:
\begin{enumerate}
	\item $D_{r_M}$ is completely degenerate: $D_{r_M} = \Lambda \, \id_{r_M}$.
	\item $R$ is orthogonal.
	\item $v_i^T v_i = 0$.
\end{enumerate}
Again, the second condition allows us to write
\begin{equation} \label{eq:paramH0}
y_i^\dagger y_i = \frac{1}{\left|f\right| \Lambda} U \, D_{\sqrt{m}}^T  \, R^\dagger v_i v_i^\dagger R D_{\sqrt{m}}  U^\dagger \equiv \frac{1}{\left|f\right| \Lambda} U \, D_{\sqrt{m}}^T \, H_0 D_{\sqrt{m}}  U^\dagger\, .
\end{equation}
Thus, we can parametrize it in terms of $H_0=R^\dagger v_i v_i^\dagger R$ which is an Hermitian positive-semidefinite $r_m \times r_m$ matrix such that $H_0 H_0^T = 0$. Again, we refer to Appendix~\ref{app:HH0} for a detailed derivation. Our results show that the number of real degrees of freedom in $H_0$, and then in $y_i^\dagger y_i$, is
\begin{equation}
  \npar_{H_0} = \frac{r_m(r_m-1)}{2} \, .
\end{equation}
Explicit expressions for the cases $r_m=3$ and $r_m=2$ are given as follows:

\subsubsection*{$\boldsymbol{r_M=r_m=3}$}

In this case $H_0$ has three free parameters, $a,b,c \in \mathbb{R}$, and can be parametrized as
\begin{equation} \label{eq:rs2N3}
H_0= \frac{1}{x}
\begin{pmatrix}
a^2+b^2      & b c+i a x  & -a c+ i b x  \\
b c -i a x   & a^2 + c^2  &  a b + i c x \\ 
-a c- i b x  & a b -i c x &  b^2 +c^2 
\end{pmatrix} \, ,
\end{equation}
where, again, $x = \sqrt{a^2 +b^2 +c^2}$.

\subsubsection*{$\boldsymbol{r_M=r_m=2}$}

In this case $H_0$ can be written in terms of only one free parameter, $a \in \mathbb{R}$, and is given by
\begin{equation} \label{eq:rs2}
H_0 = a
\begin{pmatrix}
\pm 1 & i\\
-i    & \pm 1
\end{pmatrix} \, ,
\end{equation}
where the signs of the diagonal terms must be chosen to ensure that $H_0$ is positive-semidefinite. \\

Interestingly, in both reduced scenarios, the case $r_M=r_m=2$ leads to just \textbf{\boldmath one real free parameter in the $y_i^\dagger y_i$ combination}. As one may expect, this implies definite phenomenological predictions.

\section{Phenomenological applications}
\label{sec:appli}

We show now two phenomenological applications that follow from our analysis, in two minimal Majorana neutrino mass models, one for each of the reduced scenarios discussed in the previous Section. We will focus in the $n=r_M=2$ case since this one is a more predictive scenario. However, the generalization to $n=r_M=3$ is straightforward.~\footnote{Scenarios with $n=r_M=3$ can also be regarded as minimal, in the sense that they can be parametrized in terms of $H$ or $H_0$. However, in this case the matrices $Y$ and $y_i^\dagger y_i$ depends on six and three parameters, respectively, rather than just two and one.}

\subsection{Minimal Type-I Seesaw}
\label{subsec:seesaw}

Let us consider a minimal version of the Type-I Seesaw~\cite{Minkowski:1977sc,Yanagida:1979as,Mohapatra:1979ia,GellMann:1980vs,Schechter:1980gr} with just two singlet right-handed (RH) neutrinos, $N_1$ and $N_2$. In this model $n =r_M= r_m = 2$. The pieces of the Lagrangian relevant for neutrino mass generation are
\begin{equation}
-\L = y_N \, \bar{L} \tilde{H} N + \frac{1}{2} M_N  \bar{N}^c N + \hc \, ,
\end{equation}
where $\tilde H = i \sigma_2 H^*$, with $\sigma_2$ the second Pauli matrix, $y_N$ is a $3 \times 2$ general Yukawa matrix and $M_N$ is a symmetric $2 \times 2$ matrix that can be taken to be diagonal without loss of generality. Electroweak symmetry breaking takes place in the standard way, with the Higgs developing the vacuum expectation value (VEV)
\begin{equation}
  \langle H \rangle = \frac{v}{\sqrt{2}} \begin{pmatrix}
    0 \\
    1 \end{pmatrix} \, .
\end{equation}
This induces a neutral fermion mass matrix which, in the basis $\left(\nu_L^c \, \, \, N \right)$, can be expressed as
\begin{align}
\mathcal M =
    \begin{pmatrix}
    0 & m_D\\
    m_D^T & M_N
    \end{pmatrix} \, ,
\end{align}
with $m_D = y_N \, v/\sqrt{2}$. We assume the hierarchy $\left( m_D M_N^{-1} \right)_{ij} \ll 1$ to be satisfied. This allows one to compute the light neutrinos mass matrix at leading order in seesaw expansion as
\begin{align} \label{eq:mnuseesaw}
m_\nu &= - m_D M_N^{-1} m_D^T = -\frac{1}{2} \, y_N \, v^2 \, M_N^{-1} \, y_N^T \, .
\end{align}
Therefore, comparing with Eq.~\eqref{eq:mnu} one finds the \textit{dictionary}
\begin{align}
f = -\frac{1}{2}\, , \quad Y = y_N^T \, , \quad M = v^2 \, M_N^{-1} \, .
\end{align}
Using Eq.~\eqref{eq:Yparam}, we recover the well-known Casas-Ibarra parametrization for the case of two RH neutrinos~\cite{Ibarra:2003up},
\begin{equation}
y_N^T = Y = -i \frac{\sqrt{2}}{\sqrt{v}} \, \sqrt{M_N} \, R \, D_{\sqrt{m}} \, U^\dagger \, .
\end{equation}
Since $R$ contains only two real free parameters, this minimal model is highly predictive. In particular, the decays of the heavy neutrinos are strongly constrained, in contrast to the case with three RH neutrinos. Let us consider one of the most important decay channels, $N \to \ell_\alpha W$. The decay width is given by~\cite{Atre:2009rg}
\begin{align} \label{eq:Ndecay}
\Gamma \left(N_i \to \ell_\alpha W \right) &= \frac{g^2}{64 \pi v^2} \frac{\left( 1 -\mu_W \right)^2\left( 1 +2\mu_W \right)}{\mu_W} \left| \left( R \, D{\sqrt{m}} \, U^\dagger \right)_{i \alpha } \right|^2 \nonumber \\
 &=\frac{g^2}{64 \pi v^2} \frac{\left( 1 -\mu_W \right)^2\left( 1 +2\mu_W \right)}{\mu_W} \left| \left[
\begin{pmatrix}
\sqrt{1-z^2} & z \\
-z & \sqrt{1-z^2}
\end{pmatrix}
D_{\sqrt{m}} \, U^\dagger \right]_{i \alpha } \right|^2 \, ,
\end{align}
where $g$ is the $SU(2)_L$ coupling constant, and we have defined $\mu_W = m_W^2/(M_N)_{ii}^2$, with $m_W$ the W-boson mass. The second equality corresponds to a specific parametrization of $R$ with $z = z_s + i \, z_i$, the complex free parameter. One might naively expect that, since there are only two real parameters, measuring two decay processes would be sufficient to fully determine $R$ and, consequently, predict the third decay. However, it is important to note that Eq.~\eqref{eq:Ndecay} is not injective with respect to the $R$ parameters. Therefore, observing two processes of $N_i \to \ell_\alpha W$ does not uniquely determine $z_r$ and $z_i$, as illustrated in Fig.~\ref{fig:Ndecay}. This figure has been obtained by fixing all neutrino oscillation parameters to their best-fit values for NH and IH according to the global fit~\cite{deSalas:2020pgw}. Indeed, we need three decay processes to completely determine $R$, allowing one to predict the remaining decays and other channels that depend on $y_N$. Importantly, Fig. \ref{fig:Ndecay} shows that the ratios of branching ratios (BRs) quickly approach their asymptotic values as $|z_r|$ and $|z_i|$ increase. In this asymptotic regime, the ratios become essentially fully correlated: measuring a single ratio is enough to determine all the others. In particular, for $|z_r|, |z_i| \gtrsim 2$, measuring just one process allows one to infer the remaining ratios with high confidence.

\begin{figure}[ht!]
    \centering
    \begin{subfigure}{0.49\linewidth}
    \includegraphics[width=\textwidth]{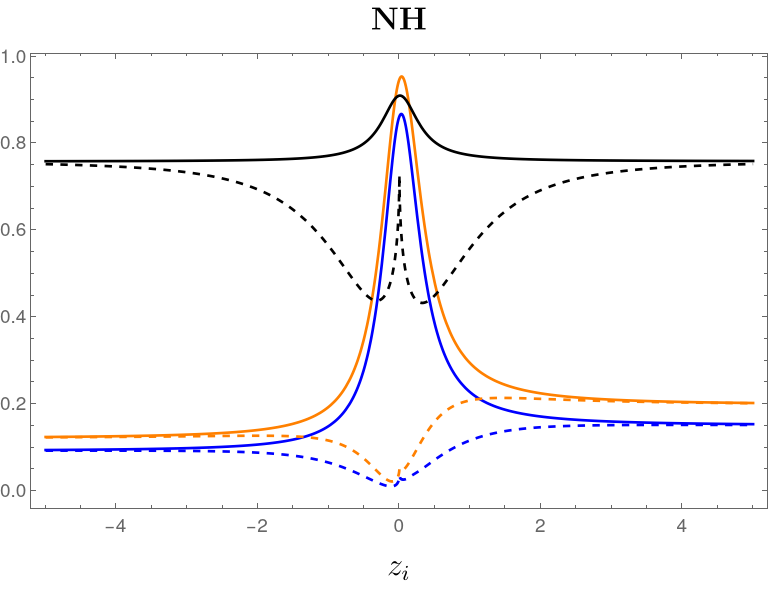}
    \label{fig:NdecayNH}
    \end{subfigure}
    \hfill
    \begin{subfigure}{0.49\linewidth}
    \includegraphics[width=\textwidth]{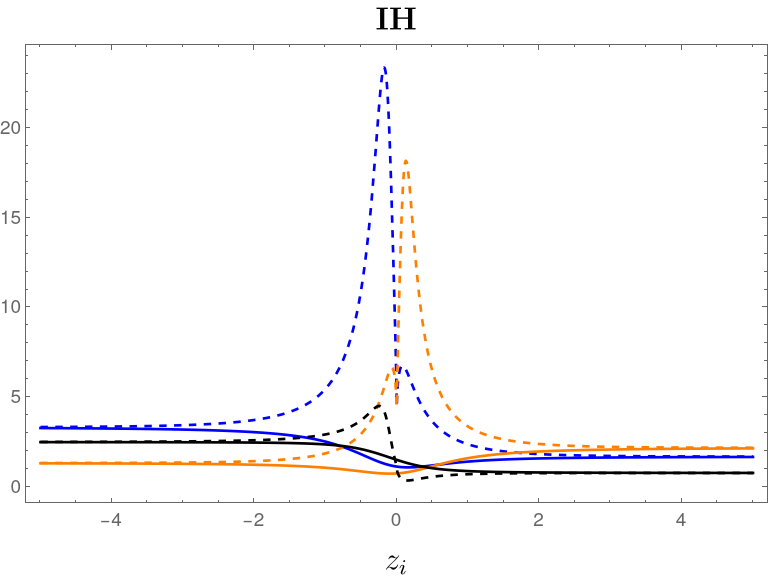}
    \label{fig:NdecayIH}
    \end{subfigure} 
    \caption{Ratios of $N_1 \to \ell_\alpha W$ branching fractions as functions of the free parameter $z_i$, for fixed values of $z_r$ ($z_r=0$ and $z_r=1$ for the solid and dashed lines, respectively), in the minimal Type-I Seesaw model. The colors indicate different ratios: $\text{BR}(N_1 \to eW)/\text{BR}(N_1 \to \mu W)$ (blue), $\text{BR}(N_1 \to eW)/\text{BR}(N_1 \to \tau W)$ (orange), and $\text{BR}(N_1 \to \tau W)/\text{BR}(N_1 \to \mu W)$ (black). Neutrino oscillation data are taken at their best-fit values for NH (left) and IH (right)~\cite{deSalas:2020pgw}.
    \label{fig:Ndecay}}
\end{figure}

On the other hand, one of the most commonly studied signature in seesaw models, and more generally in any model that induces neutrino masses, is the lepton flavor violating process $\ell_\alpha \rightarrow \ell_\beta \, \gamma$, with $\alpha \neq \beta$. At first order in $\mu_W$ and $m_\nu/m_W$ one has~\cite{Lavoura:2003xp}
\begin{equation} 
\text{BR}(\ell_\alpha \rightarrow \ell_\beta \,  \gamma) = \frac{\alpha_W^3 s_W^2}{1024 \pi^2}\left(\frac{m_{\ell_\alpha}}{m_W}\right)^4 \frac{m_{\ell_\alpha}}{\Gamma_{\ell_\alpha}} \left|\left(m_D^* M_N^{-1} (M_N^{-1})^\dagger m_D^T\right)_{\alpha \beta} \right|^2 \,. \label{eq:muegamma} 
\end{equation}
with $\alpha_W$ the weak structure constant, $s_W \equiv \sin \theta_W$ the sine of the weak mixing angle and $m_{\ell_\alpha}$ and $\Gamma_{\ell_\alpha}$ the mass and total decay rate of the charged lepton $\ell_\alpha$, respectively. This decay, in principle, depends on two real free parameters included in $m_D$. However, if we consider the case of degenerate $M_N$ we can make use of the reduced scenario 1 discussed in Sec.~\ref{sec:CIPheno}. Defining $M_N = m_N \, \id_2$ and using Eq.~\eqref{eq:paramH}, we can write
\begin{equation} 
\text{BR}(\ell_\alpha \rightarrow \ell_\beta \, \gamma) = \frac{\alpha_W^3 s_W^2}{1024 \pi^2}\left(\frac{m_{\ell_\alpha}}{m_W}\right)^4 \frac{m_{\ell_\alpha}}{\Gamma_{\ell_\alpha}} \left|\frac{\left(U D_{\sqrt{m}}^T H D_{\sqrt{m}} U^\dagger \right)_{\alpha \beta}}{m_N}\right|^2 \,, \label{eq:muegammaseesaws} 
\end{equation}
which depends on only one parameter, as shown in the previous Section. This is a very predictive scenario. In fact, the ratios among different flavor violating rates just depend on the free parameter $a$:
\begin{align}
\frac{\text{BR}(\ell_\alpha \rightarrow \ell_\beta \, \gamma)}{\text{BR}(\ell_\gamma \rightarrow \ell_\delta \, \gamma)} = \frac{m_{\ell_\alpha}^5 \, \Gamma_{\ell_\gamma}}{m_{\ell_\gamma}^5 \, \Gamma_{\ell_\alpha}} \left|\frac{\left(U D_{\sqrt{m}}^T H D_{\sqrt{m}} U^\dagger \right)_{\alpha \beta}}{\left(U D_{\sqrt{m}}^T H D_{\sqrt{m}} U^\dagger \right)_{\gamma \delta}}\right|^2 \, .
\end{align}
These ratios of BRs are shown in Fig.~\ref{fig:RatioBR}. Again, we fixed all neutrino oscillation parameters to their best-fit values for NH and IH~\cite{deSalas:2020pgw}. In this figure, an analogous discussion to that done for the $N$ decay holds. Since the branching ratio in Eq.~\eqref{eq:muegamma} is not injective in the parameter $a$, measuring a single process is not sufficient to determine the others, and at least two processes are required~\footnote{An exception occurs in the case of  $\mu \to e \gamma$ and $\tau \to e \gamma$, since they do determine each other even if the function is not injective. One can show that $|(U D_{\sqrt{m}}^T H D_{\sqrt{m}} U^\dagger)_{21}|^2= c_1 |(U D_{\sqrt{m}}^T H D_{\sqrt{m}} U^\dagger)_{31}|^2 + c_2$, with $c_1$ and $c_2$ two constants.}. However, in the asymptotic regime of the parameter space, this limitation disappears: observing just one process fully determines all the others. This Figure allows to easily rule the minimal Type-I seesaw as the explanation for neutrino mass generation. This would happen, for instance, if $\text{BR}(\mu \rightarrow e \,  \gamma) \approx \text{BR}(\tau \rightarrow e \,  \gamma) \approx \text{BR}(\tau \rightarrow \mu \,  \gamma)$ is experimentally found, since no value of the free parameter $a$ leads to this scenario.

\begin{figure}[ht!]
    \centering
    \begin{subfigure}{0.45\linewidth}
    \includegraphics[width=\textwidth]{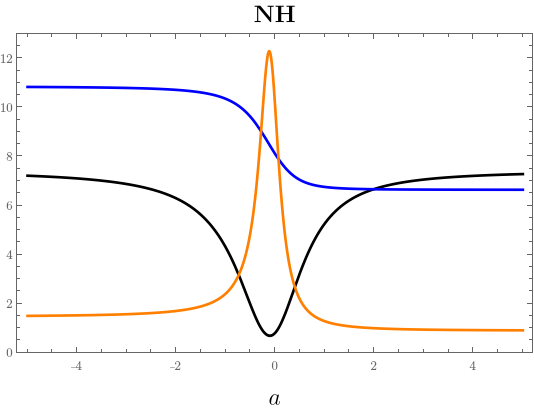}
    \label{fig:RatioBRNH}
    \end{subfigure}
    \hfill
    \begin{subfigure}{0.45\linewidth}
    \includegraphics[width=\textwidth]{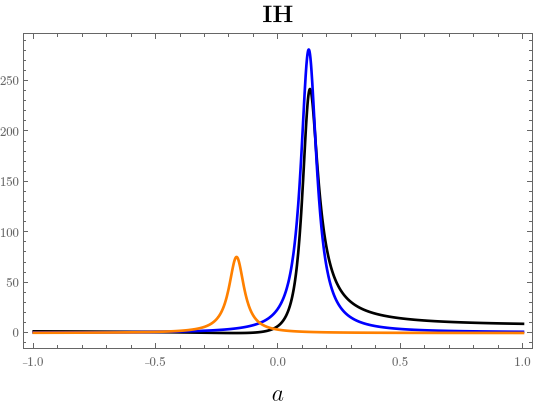}
    \label{fig:RatioBRIH}
    \end{subfigure}
     \caption{Ratios of $\ell_\alpha \to \ell_\beta \,  \gamma$ branching ratios as a function of the free parameter $a$ in the minimal Type-I Seesaw model.  $\text{BR}\left(\mu \to e \gamma\right)/\text{BR}\left(\tau \to e \gamma\right)$ in black, $\text{BR}\left(\tau \to \mu \gamma\right)/\text{BR}\left(\tau \to e \gamma\right)$ in blue and  $\text{BR}\left(\tau \to \mu \gamma\right)/\text{BR}\left(\mu \to e \gamma\right)$ in orange. Best-fit values and NH (left) and IH (right) are assumed for the neutrino oscillation data~\cite{deSalas:2020pgw}.
     \label{fig:RatioBR}}
\end{figure}

Let us remark that the results presented here are equally valid for any other model that satisfies Eq.~\eqref{eq:muegamma} and the conditions for using the $Y^\dagger Y$ parametrization. This is the case, for example, for the minimal versions of both the Inverse~\cite{Mohapatra:1986bd,Gonzalez-Garcia:1988okv} and Linear Seesaw~\cite{Akhmedov:1995ip,Akhmedov:1995vm,Malinsky:2005bi} models.

%%%%%%%%%%%%%%%%%%%%%%%%%%%%%%%%%%%%%%%%%%%%%%
%%%%%%%%%%%%%%%%%%%%%%%%%%%%%%%%%%%%%%%%%%%%%%
%
\subsection{Minimal Linear Seesaw with spontaneously broken lepton number}\label{sec:Linear}

We now consider a minimal Linear Seesaw model as studied in~\cite{CentellesChulia:2024uzv}. The SM particle content is extended with the addition of two singlet fermions, $N$ and $S$, as well as the doublet scalar $\chi$ (with the same hypercharge as the Higgs doublet) and the singlet scalar $\sigma$. The conservation of a global lepton number symmetry $U(1)_L$ is also assumed, with $L(L) = L(N) = 1$, $L(S) = -3$, $L(\chi) = -4$ and $L(\sigma) = -2$. The Yukawa Lagrangian of the model is given by
\begin{equation}
-\L = y_N \, \bar{L} \tilde{H} N + y_S \, \bar{L} \tilde{\chi} S + \lambda \, \sigma^* \bar{N}^c S +\frac{1}{2} \lambda_N \, \sigma \, \bar{N}^c N  + \hc \, ,
\end{equation}
where $\tilde \chi = i \sigma_2 \chi^*$ and $y_N$ and $y_S$ are $3 \times 1$ column matrices. We assume a vacuum characterized by
\begin{equation}
  \langle H \rangle = \frac{v}{\sqrt{2}} \begin{pmatrix}
    0 \\
    1 \end{pmatrix} \, , \quad \langle \chi \rangle = \frac{v_\chi}{\sqrt{2}} \begin{pmatrix}
    0 \\
    1 \end{pmatrix} \, , \quad \langle \sigma \rangle = \frac{v_\sigma}{\sqrt{2}}  \, .
\end{equation}
This induces the neutral fermion mass matrix
\begin{align}
\mathcal M =\begin{pmatrix}
    0 & m_D & m_L\\
    m_D^T & m_N & m_R\\
    m_L^T & m_R & 0
    \end{pmatrix} \, ,
\end{align}
where $m_D = y_N \, v/\sqrt{2}$, $m_L = y_S \, v_\chi/\sqrt{2}$, $m_N = \lambda_N \, v_\sigma/\sqrt{2}$ and $m_R = \lambda \, v_\sigma/\sqrt{2}$. If the scalars VEVs follow the hierarchy $v_\chi \ll v \ll v_\sigma$, neutrino masses are generated via a Linear Seesaw mechanism~\cite{Akhmedov:1995ip,Akhmedov:1995vm,Malinsky:2005bi}~\footnote{It is implicitly assumed that possible hierarchies in the Yukawa couplings are insufficient to overcome those in the VEVs.}
\begin{align}
  m_\nu =& -m_D \left(m_R^T \right)^{-1} m_L^T - m_L \, m_R^{-1} m_D^T \nonumber \\
  =& -\frac{1}{\sqrt{2}} \begin{pmatrix} 
y_N & y_S v_\chi/v
\end{pmatrix}
\frac{v^2}{ \lambda v_\sigma}
\begin{pmatrix} 
0 & 1\\
1 & 0
\end{pmatrix}
\begin{pmatrix}
 y_N^T \\
 y_S^T v_\chi/v
\end{pmatrix} \, . \label{eq:mnulinear}
\end{align}
Now we can compare with Eq.~\eqref{eq:mnu} and find the dictionary
\begin{align}
f = -\frac{1}{\sqrt{2}} \, , \quad Y = \begin{pmatrix}
 y_N^T \\
 y_S^T \, v_\chi/v
\end{pmatrix} \, , \quad M = \frac{v^2}{\lambda v_\sigma}
\begin{pmatrix} 
0 & 1\\
1 & 0
\end{pmatrix} \, .
\end{align}
On the other hand, the scalar sector includes an SM-like Higgs as well as the usual Goldstone bosons which provide masses to the electroweak gauge bosons. Additionally, the model contains two new heavy CP-even scalars with masses of the order $\sim v_\sigma^2$ and $\sim v_\sigma v_H/v_\chi$, one physical charged scalar with mass $\sim v_\sigma v_H/v_\chi$ and two new CP-odd scalars, one massive again with mass $\sim v_\sigma v_H/v_\chi$ and one physical massless Goldstone boson, the majoron ($J$) associated to the spontaneous breaking of $U(1)_L$.

Since the matrix $M$ has two degenerate eigenvalues, the model is a reduced scenario 2, as defined in Sec.~\ref{sec:CIPheno}. The Yukawa sector is then parametrized by just one free parameter, that again we denote by $a$, which necessarily implies that flavor observables will be predicted to follow very specific patterns. Let us first consider an exotic flavor violating decay present in this model: $\ell_\alpha \rightarrow \ell_\beta J$. The coupling between the majoron and the charged leptons is generated at one-loop and is given by~\cite{CentellesChulia:2024uzv,Herrero-Brocal:2023czw}
\begin{align} \label{eq:Jlinear}
\mathcal{L}_{\ell \ell J} = \frac{-iJ}{32\pi^2 v_\sigma} \, \bar{\ell} \, \Bigg\{\left(1 -\frac{ |\lambda_N|^2}{3 |\lambda|^2 }\right) M_\ell \textup{Tr}\left( y_N \, y_N^\dagger \right) \gamma_5 +\left( 2 -\frac{5  |\lambda_N|^2}{12  |\lambda|^2}\right) \left[ M_\ell \, y_N \, y_N^\dagger P_L - y_N y_N^\dagger \, M_\ell P_R \right] \Bigg\} \, \ell \, .
\end{align}
The coupling depends of the combination $y_N y_N^\dagger$, which can be written thanks to Eq.~\eqref{eq:paramH0} as
\begin{equation}
y_N y_N^\dagger =\sqrt{2}\,  a \, \frac{ \lambda v_\sigma}{v^2} \, \left( U D_{\sqrt{m}}^T 
\begin{pmatrix}
\pm 1 & i\\
-i    & \pm 1
\end{pmatrix}
 D_{\sqrt{m}} U^\dagger \right)^*\, .
\end{equation}
The free parameter $a$ appears as a global factor, thereby leaving the ratio between the different matrix elements of the coupling unaffected. The same conclusion holds for the branching ratio of the decay $\ell_\alpha \to \ell_\beta J$, which is given by
\begin{align}
\text{BR}\left( \ell_\alpha \to \ell_\beta J \right)&= \frac{1}{2 (8 \pi)^5} \left(\frac{m_{\ell_\alpha}}{v_\sigma}\right)^2\frac{m_{\ell_\alpha}}{\Gamma_{\ell_\alpha}} \left( 2 -\frac{5  |\lambda_N|^2}{12  |\lambda|^2}\right)^2 \left|\left(y_N y_N^\dagger\right)_{\alpha \beta} \right|^2 \nonumber \\
&= \frac{a^2 \, |\lambda|^2}{ (8 \pi)^5} \left(\frac{m_{\ell_\alpha}}{v}\right)^2\frac{m_{\ell_\alpha}}{\Gamma_{\ell_\alpha}} \left( 2 -\frac{5  |\lambda_N|^2}{12  |\lambda|^2}\right)^2 \frac{\left|\left(U D_{\sqrt{m}}^T 
\begin{pmatrix}
\pm 1 & i\\
-i    & \pm 1
\end{pmatrix}
 D_{\sqrt{m}} U^\dagger \right)_{\alpha \beta} \right|^2}{v^2} \, .
\end{align}
Again, using best-fit values for the neutrino oscillation parameters we obtain the definite predictions
\begin{align}
\frac{\text{BR}\left(\mu \to e J \right)}{\text{BR}\left(\tau \to e J \right)} \approx 2 \cdot 10^3 \, , && \frac{\text{BR}\left(\mu \to e J \right)}{\text{BR}\left(\tau \to \mu J \right)} \approx 3 \cdot 10^2 \, , && \frac{\text{BR}\left(\tau \to \mu J \right)}{\text{BR}\left(\tau \to e J \right)} \approx 6 \, ,  
\end{align}
for NH and
\begin{align}
\frac{\text{BR}\left(\mu \to e J \right)}{\text{BR}\left(\tau \to e J \right)} \approx 2 \cdot 10^3 \, , && \frac{\text{BR}\left(\mu \to e J \right)}{\text{BR}\left(\tau \to \mu J \right)} \approx 3 \cdot 10^3 \, , && \frac{\text{BR}\left(\tau \to \mu J \right)}{\text{BR}\left(\tau \to e J \right)} \approx 6 \cdot 10^{-1} \, ,  
\end{align}
for IH.\\

Let us now discuss the more conventional decay $\ell_\alpha \to \ell_\beta \, \gamma$. The new charged scalar decouples for large $v_\sigma$ and the phenomenology associated to $\ell_\alpha \to \ell_\beta \, \gamma$ remains the same as described in Sec.~\ref{subsec:seesaw}. In particular, Eq.~\eqref{eq:muegamma} holds true. Moreover, although the singlet fermions are not mass degenerate in general, in the limit $\lambda_N \ll \lambda$ they can be regarded as such, with an approximate mass $m_R$. We will then take this limit. In this case, the rate of $\ell_\alpha \to \ell_\beta \gamma$ is again described by the combination $Y^\dagger Y$. Furthermore, this matrix is related to $y_N^\dagger y_N$ by
\begin{equation} \label{eq:Ysplit}
Y^\dagger Y= y_N^* y_N^T + \frac{v_\chi^2}{v^2}y_S^* y_S^T \, .
\end{equation}
One can now parametrize each of the three matrices in the previous expression by using Eqs.~\eqref{eq:rs1} and \eqref{eq:rs2}. This leads to
\begin{equation}  \label{eq:relHH0}
\begin{pmatrix}
\sqrt{1+a_Y^2} & i a_Y\\
-i a_Y         & \sqrt{1+a_Y^2}
\end{pmatrix}
 =  a_{y_N} \, 
\begin{pmatrix}
\pm 1  & i\\
-i & \pm 1
\end{pmatrix}
+ a_{y_S} \, 
\begin{pmatrix}
\mp 1 & i\\
-i & \mp 1
\end{pmatrix} \, ,
\end{equation}
where the sign of $H_0$ has been chosen to be consistent.~\footnote{Let us recall that the signs of the diagonal elements of $H_0$ are chosen to make $H_0$ positive definite. Moreover, the off-diagonal terms of Eq.~\eqref{eq:Ysplit} impose different signs for $a_{y_N}$ and $a_{y_s}$. We have taken this into account by considering a relative sign in equation~\eqref{eq:relHH0}. Choosing the same signs would lead to a contradiction.} One can now determine the parameters $a_Y$ and $a_{y_S}$ in terms of $a_{y_N}$ as
\begin{align}
a_Y &= a_{y_N}-\frac{1}{4 a_{y_N}} \, , \\
a_{y_S}&=-\frac{1}{4 a_{y_N}} \, .
\end{align}
These results imply that the free parameter $a$ which appears in $\ell_\alpha \to \ell_\beta J$ also controls $\ell_\alpha \to \ell_\beta \gamma$. Moreover, the Linear Seesaw hierarchy $v_\chi \ll v$ implies $\frac{v_\chi^2}{v^2}y_S^* y_S^T \ll y_N^* y_N^T$, and then $a_{y_N} \gg 1$, which is nothing but $Y^\dagger Y= y_N^* y_N^T$. Thus, one finds the relation
\begin{align}
\frac{\text{BR}\left(\mu \to e J \right)}{\text{BR}\left(\mu \to e \gamma \right)} \approx \frac{|\lambda|^2}{8 \pi^3 \alpha_w^3 s_w^2}\frac{m_W^4 }{m_{\ell_\alpha}^2 v^2} \left(\frac{m_R}{v}\right)^2\approx 3.4 \cdot 10^7 \, |\lambda|^2\, \left(\frac{m_R}{v}\right)^2 \, .
\end{align}
For $\lambda \sim 1$, $m_R \sim v_\sigma \gg v$. Therefore, one expects $\text{BR}\left(\mu \to e J \right) \gg \text{BR}\left(\mu \to e \gamma \right)$.

\section{Summary}
\label{sec:summ}

Majorana neutrino mass models are among the most popular and well-motivated extensions of the SM. To determine which model correctly explains the origin of neutrino masses, we must test and potentially falsify the predictions of these models through experiments. For this purpose, it is crucial to understand the predictive power of each model.

In this work we have derived an extension of the Casas-Ibarra parametrization for the Yukawa, $Y$, which allows for an efficient counting of the free parameters in a given neutrino mass model. This parametrization provides a tool for identifying minimal models, i.e., models with very few free parameters and definite experimental predictions. Moreover, we have obtained a novel parametrization for the combination $y_i^\dagger y_i$ which is particularly convenient when confronting model predictions with experimental data. Using the previously developed parametrization for $Y$, we have shown that in certain scenarios $y_i^\dagger y_i$ manifests a reduction in the number of free parameters, which may be as low as just one.

We have finally illustrated our results with strong, testable predictions for two well-known models in the limit where our parametrization is applicable. In particular, we have considered the heavy neutrino decays $N_i \to \ell_\alpha W$ and flavor violating of charged leptons. The investigated observables offer a clear pathway to potentially rule out such models if future measurements contradict the predictions. Importantly, the predictions made for the models in Sec.~\ref{sec:appli} can be extended to other models with similar properties.

Our approach has two main limitations. First of all, models with new states beyond those required for neutrino mass generation may have additional contributions to the observables we have considered which cannot be accounted for in our analysis. This is for instance the case of the model in~\cite{Antusch:2023jok}, which is minimal according to our definitions but does not lead to definite predictions in $\ell_\alpha \to \ell_\beta \, \gamma$ due to the existence of additional non-correlated contributions. A second limitation of our approach is that the $y_i^\dagger y_i$ parametrization requires certain conditions to be applicable. However, even if a specific model does not satisfy these conditions, the predictions obtained from the parametrization can still provide important information about how far the model is from interesting theoretical limits, such as the case of degenerate heavy neutrinos. In any case, the limit where the parametrization holds serves as a powerful starting point for falsifying models that would otherwise be hard to test due to their large parameter freedom.

\section*{Acknowledgements}

The authors are grateful to Isabel Cordero-Carri\'on for carefully
reading and pointing out some missing details in the mathematical
proofs and to Martin Hirsch for fruitful discussions. Work supported
by the Spanish grants PID2023-147306NB-I00, CNS2024-154524 and
CEX2023-001292-S (MICIU/AEI/10.13039/501100011033), as well as
CIPROM/2021/054 (Generalitat Valenciana). The work of AHB is supported
by the grant No. CIACIF/2021/100 (also funded by Generalitat
Valenciana).

\appendix
\section{Equivalence with the master parametrization}
\label{app:Master}

Our analysis has taken advantage of our newly developed extended Casas-Ibarra parametrization, which simplifies the task of counting the free parameters in the Yukawa matrices. However, the same results can be obtained using the master parametrization~\cite{Cordero-Carrion:2018xre,Cordero-Carrion:2019qtu}, as we proceed to show in this Appendix. \\

Let us first of all introduce the notation and the matrices that appear in the master parametrization. For an extended discussion we refer to~\cite{Cordero-Carrion:2018xre,Cordero-Carrion:2019qtu}. We consider the general expression for the neutrino mass matrix
\begin{equation} \label{eq:mnumaster}
m_\nu = f \left( y_1^T \tilde{M} y_2 + y_2^T \tilde{M}^T y_1 \right) \, ,
\end{equation}
with $y_1$, $y_2$ and $\tilde{M}$ three $n_1 \times 3$, $n_2 \times 3$ and $n_1 \times n_2$ general matrices, respectively. First, we note that we can always rewrite Eq.~\eqref{eq:mnumaster} in the form of Eq.~\eqref{eq:mnu}:
\begin{equation} \label{eq:mnumaster2}
m_\nu = f \left( y_1^T \tilde{M} y_2 + y_2^T \tilde{M}^T y_1 \right) = f 
\begin{pmatrix} y_1 & y_2 \end{pmatrix}^T
\begin{pmatrix}
0           & \tilde{M}\\
\tilde{M}^T & 0
\end{pmatrix}
\begin{pmatrix}
y_1\\
y_2
\end{pmatrix} \, .
\end{equation}
Therefore, the Yukawa matrices of any Majorana neutrino mass model can be written in terms of neutrino oscillation parameters by means of our new extended Casas-Ibarra parametrization or with the master parametrization. Let us now concentrate on the latter. Equivalently to Eq.~\eqref{eq:Udef}, one can apply a singular value decomposition to both $m_\nu$ and $\tilde{M}$,
\begin{align}
  U^T m_\nu \, U = D_m = \bar{D}_{\sqrt{m}} \, \bar{D}_{\sqrt{m}}
\end{align}
and
\begin{align}
  V_1^T \tilde{M} V_2 = \hat{\Sigma} \, .
\end{align}
It is important to note that the definition of $\bar{D}_{\sqrt{m}}$ in~\cite{Cordero-Carrion:2018xre,Cordero-Carrion:2019qtu} differs from ours, given in Eq.~\eqref{eq:Dsqrtm},
\begin{equation} 
\bar{D}_{\sqrt{m}} = \left\{ \begin{array}{cl}
\textup{diag}\left( \sqrt{m_1},\sqrt{m_2},\sqrt{m_3} \right) \,, \quad & \textup{if} \, \,\, r_m = 3 \, , \\
& \\
P \cdot \textup{diag}\left( \sqrt{v},\sqrt{m_2},\sqrt{m_3} \right) \cdot P \,, \quad & \textup{if} \, \,\, r_m = 2 \, .
\end{array} \right.
\end{equation}
The master parametrization is then given by
\begin{align} \label{eq:masterparam}
y_1 &= \frac{1}{\sqrt{2 f}} V_1^\dagger 
\begin{pmatrix}
\Sigma^{-1/2} W A\\
X_1 \\
X_2
\end{pmatrix}
\bar{D}_{\sqrt{m}} U^\dagger \, , \\
\nonumber \\
y_2 &= \frac{1}{\sqrt{2 f}} V_2^\dagger 
\begin{pmatrix}
\Sigma^{-1/2} \hat{W}^* \hat{B}\\
X_3 \\
\end{pmatrix}
\bar{D}_{\sqrt{m}} U^\dagger \, , \label{eq:masterparam2}
\end{align}
where $\Sigma$ is a diagonal $r_{\tilde{M}} \times r_{\tilde{M}}$ matrix containing the $r_{\tilde{M}}$ non-zero singular values of $\tilde{M}$. $X_{1,2,3}$ are, respectively, $(n_2-r_{\tilde{M}})\times 3$, $(n_1-n_2)\times 3$ and $(n_2-r_{\tilde{M}})\times 3$ arbitrary complex matrices. $\hat{W}$ is an $r_{\tilde{M}} \times r_{\tilde{M}}$ unitary matrix formed by
\begin{equation}
\hat{W} = \begin {pmatrix} W & \bar{W} \end{pmatrix}
\end{equation}
with $W$ an $r_{\tilde{M}} \times r$ and $\bar{W}$ an $r_{\tilde{M}} \times (r_{\tilde{M}}-r)$ complex matrices. Here $r = \text{rank}(W)$. $A$ is given by
\begin{equation}
A = T C_1
\end{equation}
with $T$ an upper-triangular $r \times r$ invertible matrix with positive real values in the diagonal, and $C_1$ is an $r \times 3$ matrix. Finally, $\hat{B}$ is defined by
\begin{equation}
\hat{B} = \begin{pmatrix}
B \\
\bar{B}
\end{pmatrix}
\end{equation}
where $\bar{B}$ is an arbitrary $(r_{\tilde{M}}-r)\times 3$ complex matric and $B$ an $r \times 3$ complex matrix written as function of $T$, $C_1$, $C_2$ and $K$,
\begin{equation}
B = \left( T^T \right)^{-1} \left[C_1 C_2 + K C_1 \right]
\end{equation}
with $K$ is an antisymmetric $r \times r$ matrix and $C_2$ a $3 \times 3$ matrix. The matrices $C_1$ and $C_2$ take specific forms that depend on the values of $r_m$ and $r$ and can be found in~\cite{Cordero-Carrion:2018xre,Cordero-Carrion:2019qtu}.

%%%%%%%%%%%%%%%%%%%%%%%%%%%%%%%%%%%%%
\subsection{Finding the extended Casas-Ibarra from the master parametrization}

We will now show that the extended Casas-Ibarra parametrization can be recovered from the master parametrization. In order to do so, we must first unify the notation between \( \bar{D}_{\sqrt{m}} \) and \( D_{\sqrt{m}} \). To do this, we rewrite Eq.~\eqref{eq:Yparam} as
\begin{equation} \label{eq:Yparamapp}
Y = \frac{1}{\sqrt{f}}
V^\dagger  \, \begin{pmatrix}
 \, \sqrt{D_{r_M}}^{-1}\, R  \, O \,  \bar{D}_{\sqrt{m}} \, U^\dagger   \\
X
\end{pmatrix}  \equiv  \frac{1}{\sqrt{f}}
V^\dagger  \, \begin{pmatrix}
 \, \sqrt{D_{r_M}}^{-1}\, R   \,O    \\
X
\end{pmatrix}\, \bar{D}_{\sqrt{m}} \, U^\dagger \, ,
\end{equation}
where we have defined 
\begin{equation} \label{eq:Yparamapp2}
O = \left\{ \begin{array}{cl}
\id_3 \,, \quad & \textup{if} \, \, r_m = 3 \, , \\
& \\
\begin{pmatrix}
0 & 1 & 0\\
0 & 0 & 1
\end{pmatrix} \, P \,, \quad & \textup{if} \, \, r_m = 2 \, ,
\end{array} \right.
\end{equation}
and we have made use of the freedom of $X$ to redefine it as $X\equiv X \bar{D}_{\sqrt{m}} \, U^\dagger$. Next, we consider
\begin{align}
 y_1 = y_2 = \frac{1}{\sqrt{2}} Y \, , && M = \tilde{M} \, .
\end{align}
Since now $\tilde{M}=M$ is symmetric, $V_1 = V_2 \equiv V$ and $ \sqrt{D_{r_M}}^{-1}  \equiv \Sigma^{-1/2}$. On the other hand, from $y_1 = y_2$ we obtain~\footnote{This is a choice that we can always make without loss of generality, since Eq.~\eqref{eq:mnumaster2} can be used to reparametrize the Yukawas.}
\begin{align}
W A &= \hat{W}^* \hat{B} \, , \label{eq:WAWB} \\
\begin{pmatrix}
X_1 \\
X_2
\end{pmatrix} & = X_3 \equiv X \, . \label{eq:X1X2X3}
\end{align}
From Eq.~\eqref{eq:WAWB} one finds
\begin{align}
W^T \, W A &= B \, , \label{eq:WAWB1} \\
\bar{W}^T \, W A &= \bar{B} \, . \label{eq:WAWB2}
\end{align}
Now, using the relation
\begin{equation} \label{eq:ABBA}
A^T B + B^T A = \left\{ \begin{array}{cl}
\id_3 \,, \quad & \textup{if} \, \, r_m = 3 \, , \\
& \\
P 
\begin{pmatrix} 
0 & 0 & 0 \\
0 & 1 & 0 \\
0 & 0 & 1
\end{pmatrix}
P \,, \quad & \textup{if} \, \, r_m = 2 \, .
\end{array} \right. \, ,
\end{equation}

we find
\begin{equation} 
A^T W^T W A = \left\{ \begin{array}{cl}
\id_3 \,, \quad & \textup{if} \, \, r_m = 3 \, , \\
& \\
P 
\begin{pmatrix} 
0 & 0 & 0 \\
0 & 1 & 0 \\
0 & 0 & 1
\end{pmatrix}
P \,, \quad & \textup{if} \, \, r_m = 2 \, ,
\end{array} \right.
\end{equation}

and we can identify $\tilde{R} \equiv WA$. This allows us to write the master parametrization as
\begin{equation} 
y_1 \equiv \frac{1}{\sqrt{2}} Y
= \frac{1}{\sqrt{2 f}} \,
V^\dagger  \, \begin{pmatrix}
 \, \Sigma^{-1/2} \tilde{R} \\
X
\end{pmatrix} \, \bar{D}_{\sqrt{m}} \, U^\dagger \, .
\end{equation}
Finally, it is trivial to identify $\tilde{R} \equiv R O$ and, introducing $\bar{D}_{\sqrt{m}} U^\dagger$ into $X$, we find
\begin{align} 
Y &= \frac{1}{\sqrt{f}}
V^\dagger  \, \begin{pmatrix}
 \,\Sigma^{-1/2}\,R \, O \, \bar{D}_{\sqrt{m}} \, U^\dagger \\
X
\end{pmatrix} =  \frac{1}{\sqrt{f}} \,
V^\dagger  \, \begin{pmatrix}
 \,\Sigma^{-1/2} \,R \,  D_{\sqrt{m}} \, U^\dagger\\
X
\end{pmatrix} \, ,
\end{align}
which is the extended Casas-Ibarra parametrization.

%%%%%%%%%%%%%%%%%%%%%%%%%%%%%%%%%%%%%
\subsection{Finding the master parametrization from the extended Casas-Ibarra}

Now, we show that the master parametrization can be derived from the extended Casas-Ibarra parametrization. We begin by rewriting Eq.~\eqref{eq:mnumaster} as Eq.~\eqref{eq:mnu} by decomposing $Y$ and $M$
\begin{equation}
m_\nu = f \, Y^T M Y = f 
\begin{pmatrix} y_1 & y_2 \end{pmatrix}^T
\begin{pmatrix}
0           & \tilde{M}\\
\tilde{M}^T & 0
\end{pmatrix}
\begin{pmatrix}
y_1\\
y_2
\end{pmatrix} \, .
\end{equation}
Again, $\tilde{M}$ and $y_i$ are $n_1 \times n_2$ and $n_i \times 3$ matrices, respectively, with $n_1 + n_2 = n$. In this case, the singular value decomposition of $M$ reads
\begin{equation}
D_M = V^* M V^\dagger = v_1^* \tilde{M} v_2^\dagger +v_2^* \tilde{M}^T v_1^\dagger \, ,
\end{equation} 
with $v_i$ the $n \times n_i$ subblock of $V$ satisfying $v_i^\dagger v_j = \delta_{ij} \id_{n}$. Since $\tilde{M}$ is an $n_1 \times n_2$ matrix, it admits a singular value decomposition in terms of two unitary matrices $\tilde{v}_1$ and $\tilde{v}_2$. Therefore, it is easy to check that one can write
\begin{align}
v_i &= \mathcal{V}_i \tilde{v}_i \, , &&
\mathcal{V}_i = \frac{1}{\sqrt{2}}
\begin{pmatrix}
\id_{r_{\tilde{M}}}       & 0_{r_{\tilde{M}} \times (n_i -r_{\tilde{M}})}\\
\pm i \, \id_{r_{\tilde{M}}} & 0_{r_{\tilde{M}} \times (n_i -r_{\tilde{M}})}\\
0_{(n-2r_{\tilde{M}}) \times r_{\tilde{M}}} & u_i
\end{pmatrix} \, ,
\end{align}
where we choose the + (-) sign for $\mathcal{V}_1$ ($\mathcal{V}_2$), respectively. The $u_i$ matrix is $(n-2r_{\tilde{M}}) \times (n_i -r_{\tilde{M}})$ and verifies $u_i^\dagger u_i = \id_{ n_i -r_{\tilde{M}}}$. With these definitions we obtain
\begin{equation}
 V^* M V^\dagger = D_M = 
 \begin{pmatrix}
 \Sigma & 0_{r_{\tilde{M}} \times r_{\tilde{M}}}     & 0_{r_{\tilde{M}} \times (n- 2r_{\tilde{M}})}\\
 0_{r_{\tilde{M}} \times r_{\tilde{M}}}     & \Sigma & 0_{r_{\tilde{M}} \times (n- 2r_{\tilde{M}})}\\
 0_{(n- 2r_{\tilde{M}}) \times r_{\tilde{M}}}      &  0_{(n- 2r_{\tilde{M}}) \times r_{\tilde{M}}}      &  0_{(n- 2r_{\tilde{M}}) \times (n- 2r_{\tilde{M}})}
 \end{pmatrix}
 \, ,
\end{equation} 
with $\Sigma$ a diagonal matrix containing the $r_{\tilde{M}}$ non-vanishing singular values of $\tilde{M}$. Finally, decomposing the $R \, O$ matrix introduced in Eqs.~\eqref{eq:Yparamapp} and \eqref{eq:Yparamapp2} into $r_{\tilde{M}} \times 1$ subblocks,
\begin{equation}
R \, O = 
\begin{pmatrix}
r_{11} & r_{12} & r_{13} \\
r_{21} & r_{22} & r_{23}
\end{pmatrix} \, ,
\end{equation}
and using $v_i =\mathcal{V}_i \tilde{v}_i$ in Eq.~\eqref{eq:Yparamapp} we find
\begin{align}
y_i =& \frac{\tilde{v}_i^\dagger \mathcal{V}_i^\dagger}{\sqrt{f}} \,
 \, \begin{pmatrix}
 \, \sqrt{D_{r_M}}^{-1} R \,O \\
X
 \end{pmatrix} \,  \bar{D}_{\sqrt{m}} \, U^\dagger \nonumber \\
 =& \frac{\tilde{v}_i^\dagger}{\sqrt{2f}} \,
 \, \begin{pmatrix}
 \, \Sigma^{-1/2} \begin{pmatrix} r_{11} \mp i r_{21} & r_{12} \mp i r_{22} & r_{13} \mp i r_{23} \end{pmatrix} \\
u_i^\dagger X
\end{pmatrix} \,  \bar{D}_{\sqrt{m}} \, U^\dagger \nonumber \\
\nonumber\\
\equiv& \frac{1}{\sqrt{2f}} \,
\tilde{v}_i^\dagger \, \begin{pmatrix}
 \, \Sigma^{-1/2} \bar{\bar{y}}_i \\
u_i^\dagger X
\end{pmatrix} \,  \bar{D}_{\sqrt{m}} \, U^\dagger \, .
\end{align} 
It is now straightforward to show that 
\begin{equation}
\left(\bar{\bar{y}}_1\right)^T \bar{\bar{y}}_2 + \left(\bar{\bar{y}}_2\right)^T \bar{\bar{y}}_1= \left( R \, O \right)^T \left( R \, O\right) = 2 \left\{ \begin{array}{cl}
\id_3 \,, \quad & \textup{if} \, \, r_m = 3 \, , \\
& \\
P 
\begin{pmatrix} 
0 & 0 & 0 \\
0 & 1 & 0 \\
0 & 0 & 1
\end{pmatrix}
P \,, \quad & \textup{if} \, \, r_m = 2 \, .
\end{array} \right.
\end{equation}
These matrices are identical to those obtained by the authors in the proof of the master parametrization. As they show, they can be expressed as functions of the matrices defined at the beginning of this Appendix, $\bar{\bar{y}}_1 = W A$ and $\bar{\bar{y}}_2 = \hat{W}^* \hat{B}$. Hence, using $u_i^\dagger X \equiv Z_i$, which is an $(n_i - r_{\tilde{M}})\times 3$ general matrix, we find the expressions for $y_1$ and $y_2$
\begin{align}
y_1 &= \frac{1}{\sqrt{2 f}} \tilde{v}_1^\dagger 
\begin{pmatrix}
\Sigma^{-1/2} W A\\
Z_1 
\end{pmatrix}
\bar{D}_{\sqrt{m}} U^\dagger \, , \\
\nonumber \\
y_2 &= \frac{1}{\sqrt{2 f}} \tilde{v}_2^\dagger 
\begin{pmatrix}
\Sigma^{-1/2} \hat{W}^* \hat{B}\\
Z_2
\end{pmatrix}
\bar{D}_{\sqrt{m}} U^\dagger \, ,
\end{align}
which are completely equivalent to Eqs.~\eqref{eq:masterparam} and \eqref{eq:masterparam2}.

%%%%%%%%%%%%%%%%%%%%%%
\subsection{Counting free parameters in the master parametrization}

In the master parametrization, the free parameters of a particular model are given by
\begin{align}
  \npar_{\text{free}} &= \npar_{X_1} + \npar_{X_2} + \npar_{X_3} + \npar_{T} + \npar_{K} + \npar_{\bar{B}} + \npar_{W} + \npar_{C_1} -\npar_{\textup{extra}} \nonumber \\
  &= 6 (n_1 + n_2 - r_{\tilde{M}}) + r (r +2r_{\tilde{M}} -7) + \npar_{C_1}-\npar_{\textup{extra}} \, ,
\end{align}
where $\npar_{C_1}$ is case-dependent, as shown in~\cite{Cordero-Carrion:2018xre,Cordero-Carrion:2019qtu}. We can recover the free parameters of the extended Casas-Ibarra parametrization by eliminating the redundant free parameters arising from $y_1 = y_2$. We have already seen that this condition is equivalent to imposing Eqs.~\eqref{eq:X1X2X3}, \eqref{eq:WAWB1} and \eqref{eq:WAWB2}. Eqs.~\eqref{eq:X1X2X3} and \eqref{eq:WAWB2} completely determine the free parameters of $\bar{B}$ and $X_3$, imposing thus $6(r_{\tilde{M}}-r)$  and $6(n_2 -r_{\tilde{M}})$ restrictions, respectively. Eq.~\eqref{eq:WAWB1} is trickier, since $A$ depends on the same matrices as $B$. Therefore, we write
\begin{align}
W^T W T C_1 = \left( T^T \right)^{-1} \left[C_1 C_2 + K C_1 \right] \, \Rightarrow \, C_1^T T^T W^T W T C_1 = C_1^T \left[C_1 C_2 + K C_1 \right] \, .
\end{align}
From the symmetry of the left-hand side of this equation it is straightforward to find
\begin{equation}
K = \frac{1}{2} \left(C_1^{-1}\right)^{-1} \left[ C_2\, , \, C_1^T C_1 \right] C_1^{-1} \, ,
\end{equation}
which fixes the $r(r-1)$ parameters of $K$. Moreover, one can write
\begin{align}
W^T W  = B A^{-1} \, ,
\end{align}
imposing $r(r+1)$ conditions since $W^T W$ is an $r \times r$ symmetric matrix. Finally, Eq.~\eqref{eq:ABBA} also leads to
\begin{equation} 
C_1^T T^T W^T W T C_1 = \left\{ \begin{array}{cl}
\id_3 \,, \quad & \textup{if} \, \, r_m = 3 \, , \\
& \\
P 
\begin{pmatrix} 
0 & 0 & 0 \\
0 & 1 & 0 \\
0 & 0 & 1
\end{pmatrix}
P \,, \quad & \textup{if} \, \, r_m = 2 \, .
\end{array} \right.
\end{equation}

which can be used to fix the free parameters of the $C_1$ matrix. Therefore, after this discussion, the free parameters of the master parametrization with the identification $y_1 = y_2$ reduce to
\begin{align}
\npar_{\text{free}} \equiv \npar_Y = 6 \, n_1 + 2r_{\tilde{M}} (r-3) -r(r+1) \, .
\end{align}
Finally, let us note that if $y_1= y_2$ then $r = r_m$ and hence
matches perfectly what we found in Section~\ref{subsec:freeparam}.

\section{Mathematical proofs}
\label{app:HH0}

In the following we provide proofs of some important theorems that elucidate the properties of the matrices $H$ and $H_0$ introduced in Sec.~\ref{sec:CIPheno}. We begin by presenting some lemmas. We consider a general Hermitian positive-semidefinite (PSD) quasi-orthogonal (QO) matrix, i.e.
\begin{equation}
\left\{ P \in \mathbb{C}^{n \times n} \;\middle|\; P= P^\dagger,\; v^\dagger P v \geq 0\; \forall v \in \mathbb{C}^n \setminus \{0\},\; \text{and}\; P P^T = P^T P = \alpha \, \id_n, \quad \alpha \in \mathbb{R}_{\geq 0} \right\} \, .
\end{equation}

\begin{lemma} \label{lemma1}\footnote{This, together with Lemma~\ref{lemma3}, are already well-known results in algebra. We include them here for completeness.}
  \textit{Let $P \in \mathbb{C}^{n \times n}$ be a PSD matrix of rank $p$. Then there exists a full-rank matrix $B \in \mathbb{C}^{p \times n}$ such that $P = B^\dagger B$.}
\end{lemma}

\begin{proof}
Since $P$ is Hermitian, it can be diagonalized by a unitary matrix,
\begin{equation} \label{eq:Cholesky}
P = V^\dagger D V =\left( \sqrt{D} V \right)^\dagger \sqrt{D} V \, ,
\end{equation}
where $D$ is a diagonal matrix. In addition, since $P$ is also PSD, $D$ is non-negative and we define $\sqrt{D}$, with $\sqrt{D}_{ii} =\sqrt{D_{ii}}$. Eq.~\eqref{eq:Cholesky} shows that a PSD matrix can always be written as the product of a matrix with its conjugate transpose. Finally, to show that we can choose this matrix as a $p\times n$ matrix, it is enough to observe that the columns of $V$ associated with the zero eigenvalues of $P$ do not contribute. Thus, we can write
\begin{equation}
P = \left( \sqrt{D} V \right)^\dagger \sqrt{D} V = \left(\sqrt{D_p} \, V_p \right)^\dagger \sqrt{D_p} \, V_p \, ,
\end{equation}
where $\sqrt{D_p}$ is the submatrix of $\sqrt{D}$ containing the $p$ positive eigenvalues, and $V_p$ is the corresponding submatrix of $V$. This concludes the proof of the lemma.
\end{proof}

\begin{lemma} \label{lemma2}
  \textit{Let $P \in \mathbb{C}^{n \times n}$ be a matrix with rank $p$. This matrix is PSD-QO if and only if it can be written as}
\begin{equation}
P = R^\dagger C^\dagger C R,
\end{equation}
\textit{with $R$ an $n \times n$ complex orthogonal matrix and $C$ a full-rank $p \times n$ matrix satisfying}
\begin{equation} \label{eq:Ccond}
C^\dagger C C^T C^* = \alpha \, \id_n \, .
\end{equation}
\end{lemma}

\begin{proof}
\phantom{This text will be invisible} \\
\noindent $\left( \Leftarrow \right)$ \\
If $P= R^\dagger C^\dagger C R$, with $R R^T = R^T R = \id_n$ and $C$ a full-rank $p \times n$ matrix satisfying Eq. \eqref{eq:Ccond}, then $P$ is PSD by construction. Moreover,
\begin{equation}
P P^T = R^\dagger C^\dagger C R R^T C^T C^* R^* = R^\dagger C^\dagger C C^T C^* R^* = \alpha R^\dagger R^* = \alpha \, \id_n \, ,
\end{equation}
which completes the proof in this direction.\\

\noindent $\left( \Rightarrow \right)$ \\
According to Lemma~\ref{lemma1}, since $P$ is PSD, we can write $P = B^\dagger B$ with $B$ a full-rank $p \times n$ matrix. Being $B$ a full-ranked matrix allows us to factorize $B = C D$, where $C$ is full-rank $p \times n$ and $D$ is invertible $n \times n$. Choosing $D$ orthogonal ($D \equiv R$), the QO condition implies
\begin{equation}
P P^T = D^\dagger C^\dagger C D D^T C^T C^* D = \alpha \, \id_n \, .
\end{equation}
Since $D$ is orthogonal, this yields the condition on $C$:
\begin{equation}
C^\dagger C C^T C^* = \alpha \, \id_n \, .
\end{equation}
Thus,
\begin{equation}
P = R^\dagger C^\dagger C R,
\end{equation}
with $R$ orthogonal and $C$ satisfying the above condition.
\end{proof}

\begin{lemma} \label{lemma3}
  \textit{Let $P \in \mathbb{C}^{n \times n}$ be a matrix such that $P P^T = 0_{n \times n}$, then rank $P \leq n/2$.}
\end{lemma}

\begin{proof}
A matrix $P$ satisfying $P P^T = 0_{n \times n}$ is composed of $ 1\times n$ row vectors $p_i$ such that $p_i p_j^T = 0$ $\forall i, j$. Therefore, it suffices to show that in $\mathbb{C}^n$ there exist at most $n/2$ linearly independent vectors such that $p_i p_j^T = 0$. Let us define a bilinear form $B:\mathbb{C}^n \times \mathbb{C}^n \rightarrow \mathbb{C}^n $ by $B(v_i, v_j) = v_i v_j^T$. Let $U \in \mathbb{C}^n$ be the subspace of vectors that are mutually orthogonal with respect to this bilinear form, i.e., $B(u_i, u_j)= 0$ $\forall u_i,u_j \in U$. Then, by definition, $U = U^\perp$. Since $U$ is a subspace of $\mathbb{C}^n$, we have
\begin{equation}
\text{dim}(U)+\text{dim}(U^\perp) \leq \text{dim}(\mathbb{C}^n) = n \, ,
\end{equation}
But as $U = U^\perp$, it follows that
\begin{equation}
\text{dim}(U)\leq \frac{n}{2} \, .
\end{equation}
\end{proof}

\begin{theorem} \label{theorem1}
  \textit{Let $H$ be an $n \times n$ positive-definite matrix, then it is orthogonal ($H H^T = H^T H = \id_n$) if and only if it can be written as}
\begin{equation}
H = R^\dagger R,
\end{equation}
\textit{with $R$ a complex orthogonal matrix.}
\end{theorem}

\begin{proof}
  The proof of this Theorem makes use of Lemma~\ref{lemma2}, particularized with $\alpha = 1$ and $p = n$. \\

\noindent $\left( \Leftarrow \right)$ \\
Already proved by Lemma~\ref{lemma2}, taking $C = \id_n$. \\

\noindent $\left( \Rightarrow \right)$ \\
According to Lemma~\ref{lemma2}, $H$ can be expressed as $R^\dagger C^\dagger C R$, with $C$ satisfying Eq.~\eqref{eq:Ccond}. In this case, $C$ is invertible and we obtain
\begin{equation}
\left(C C^T\right)^\dagger = \left(C C^T\right)^{-1} \,.
\end{equation}
This shows that $C C^T $ is a symmetric unitary matrix. Such a matrix admits a Takagi decomposition with all singular values equal to one. Hence,
\begin{equation}
U^\dagger C (U^\dagger C)^T = \id_n,
\end{equation}
for some unitary $U$. Defining $\tilde{R} = U^\dagger C$, an orthogonal matrix, we get
\begin{equation}
H = R^\dagger \tilde{R}^\dagger U^\dagger U \tilde{R} R =R^\dagger \tilde{R}^\dagger \tilde{R} R = \mathcal{R}^\dagger \mathcal{R}
\end{equation}
with $\mathcal{R} = \tilde{R} R$ orthogonal, proving the statement.
\end{proof}

\begin{theorem} \label{theorem2}
  \textit{Let $H_0$ be an $n \times n$ PSD matrix, then it is isotropic if and only if}
\begin{equation}
H_0 = R^\dagger O^\dagger O R,
\end{equation}
\textit{where $R$ is an $n \times n$ complex orthogonal matrix and $O$ is an $m \times n$ complex isotropic matrix ($O O^T = 0$).}
\end{theorem}

\begin{proof}
  The proof of this Theorem makes use of Lemma~\ref{lemma2}, particularized with $\alpha = 0$. \\

\noindent $\left( \Leftarrow \right)$ \\
Direct from Lemma~\ref{lemma2}, taking $C = O$. \\

\noindent $\left( \Rightarrow \right)$ \\
According to Lemma~\ref{lemma2}, $H_0$ can be expressed as $R^\dagger C^\dagger C R$, with $C$ satisfying Eq.~\eqref{eq:Ccond}. Since $C$ is a full-ranked $p \times n$ matrix, it admits a right inverse, and Eq.~\eqref{eq:Ccond}
implies
\begin{equation}
C C^T = 0_{p \times p} \,,
\end{equation}
so $C$ is isotropic. Therefore, simply renaming $C \equiv O$, one has
\begin{equation}
H_0 = R^\dagger O^\dagger O R,
\end{equation}
with $O$ an $m \times n$ isotropic matrix.
\end{proof}

Particularly useful for us are the cases $n=3$ and $n=2$, where rank $(H_0) \equiv r_{H_0} = 1$ (since $r_{H_0} \leq n/2$ due to Lemma~\ref{lemma3}) and $C$ is a single vector. In these scenarios, we can define $v = 1/\sqrt{||O||} \, O \equiv k \, O$ and write
\begin{equation}
H_0 = \frac{1}{k^2} \, R^\dagger v^\dagger v R \, ,
\end{equation}
with $k$ a constant factor.

\begin{theorem} \label{theorem3}
  \textit{Let $P$ be an $n \times n$ PSD-QO matrix of rank $p$, then $P$ is uniquely determined by a real antisymmetric matrix.}
\end{theorem}

\begin{proof}
We can decompose $P = S + i A$ with $S=S^T$ and $A= -A^T$ being two real matrices. From the QO condition,
\begin{equation}
\alpha \, \id_n = P P^T = (S + i A)(S^T + i A^T) = S^2 + A^2 + i(AS - SA)\, .
\end{equation}
Applying the equality to both real and imaginary parts, we obtain
\begin{align}
S^2 &= \alpha \, \id_n - A^2 \, , \label{eq:SA2P} \\
[A, S] &= 0 \, . \label{eq:ASconmP}
\end{align}
Since $A$ is antisymmetric, its eigenvalues are pure imaginary (or zero) and those of $-A^2$ are positive (or zero). Therefore, the right-hand side of Eq.~\eqref{eq:SA2P} is a real PSD matrix which then admits a unique PSD square root. Moreover, $P$ being PSD implies the positive-semidefiniteness of $S$ when applying the definition for a real vector, $u \in \mathbb{R}$,
\begin{equation}
u^T P u = u^T P u + i u^T A u =  u^T P u \geq 0 \, .
\end{equation}
Consequently, $S$ in uniquely defined by $A$. We can find $S$ as follows. Since $\alpha \,\id_n - A^2$ is a real symmetric matrix it can be diagonalized by an orthogonal matrix, $Q$,
\begin{equation}
Q \left( \alpha\, \id_n - A^2 \right) Q^T =  D^2 \, ,
\end{equation}
where $D$ is a real non-negative diagonal matrix. Then
\begin{equation}
S^2 = \alpha \,\id_n - A^2 = Q^T  D^2 Q = Q^T D Q Q^T D Q
\end{equation}
and
\begin{equation}
S =  Q^T D Q \, ,
\end{equation}
where $D$ is the non-negative square root of $D^2$: $D_{ii} = \sqrt{D^2_{ii}}$. Therefore, constructing $P$ in this way ensures that $P$ is PSD. Furthermore, it is trivial to show that Eq.~\eqref{eq:ASconmP} holds. From Eq.~\eqref{eq:SA2P} $[S^2,A]=0$ and this is true in any basis. Thus, in the basis where $S^2$ is diagonal we find
\begin{equation}
D^2 Q A Q^T - Q A Q^T D^2 = \left( Q A Q^T \right)_{ij} \left( D^2_{ii}- D^2_{jj} \right) = 0 \, .
\end{equation}
Therefore, in this basis the conmutator $C=[S,A]$ will be given by
\begin{equation}
Q C Q^T = D Q A Q^T - Q A Q^T D = \left( Q A Q^T \right)_{ij} \left( D_{ii}- D_{jj} \right) = 0 \, ,
\end{equation}
and $[S,A]= 0$ in any basis. From all this, it follows that $S$, and thus $P$, is univoquely determined by $A$.
\end{proof}

This Theorem implies that the degrees of freedom of $P$ are those of a real $n \times n$ antisymmetric matrix ,
\begin{equation}
  \npar_{P} = \frac{n(n-1)}{2} \, .
\end{equation}

Finally, considering the cases $\alpha = 1 \, , 0$ and $n= 3 \, , 2$ we find the parametrizations for the matrices $H$ and $H_0$ given in Sec.~\ref{sec:CIPheno}. For $H$ we obtain
\begin{align}
H &= \frac{1}{x^2}
\begin{pmatrix}
c^2 + \left(a^2 + b^2\right)\sqrt{1 + x^2} & 
a b \left(-1 + \sqrt{1 + x^2})\right)+i a x^2  & 
a c\left(1-\sqrt{1 + x^2} \right)+ i b x^2  \\
a b  \left(-1 + \sqrt{1 + x^2}\right)-i a x^2 & 
b^2 + \left(a^2 + c^2\right)\sqrt{1 + x^2} & 
a b \left(-1 + \sqrt{1 + x^2}\right)+i c x^2  \\
a c\left(1-\sqrt{1 + x^2} \right) - i b x^2  & 
a b \left(-1 + \sqrt{1 + x^2}\right)- i c x^2 & 
a^2 + \left(b^2 + c^2\right)\sqrt{1 + x^2}
\end{pmatrix}
\end{align}
and
\begin{align}
H &= 
\begin{pmatrix}
\sqrt{1+a^2} & i a\\
-i a         & \sqrt{1+a^2}
\end{pmatrix} \, ,
\end{align}
while $H_0$ is given by,
\begin{align}
H_0= \frac{1}{x}
\begin{pmatrix}
a^2+b^2      & b c+i a x  & -a c+ i b x  \\
b c -i a x   & a^2 + c^2  &  a b + i c x \\ 
-a c+ i b x  & a b -i c x &  b^2 +c^2 
\end{pmatrix}
&& \text{and} &&
H_0 = a
\begin{pmatrix}
\pm 1 & i\\
-i    & \pm 1
\end{pmatrix}\, .
\end{align}
In the $n=2$ case, the sign in $H_0$ must be chosen according to the sign of $a$ in order to ensure positive-semidefiniteness.

\bibliographystyle{JHEP}
\bibliography{refs}

\providecommand{\href}[2]{#2}\begingroup\raggedright\begin{thebibliography}{10}

\bibitem{deSalas:2020pgw}
P.~F. de~Salas, D.~V. Forero, S.~Gariazzo, P.~Mart{\'\i}nez-Mirav{\'e},
  O.~Mena, C.~A. Ternes et~al., \emph{{2020 global reassessment of the neutrino
  oscillation picture}},
  \href{http://dx.doi.org/10.1007/JHEP02(2021)071}{\emph{JHEP} {\bf 02} (2021)
  071}, [\href{http://arxiv.org/abs/2006.11237}{{\tt 2006.11237}}].

\bibitem{Boucenna:2014zba}
S.~M. Boucenna, S.~Morisi and J.~W.~F. Valle, \emph{{The low-scale approach to
  neutrino masses}}, \href{http://dx.doi.org/10.1155/2014/831598}{\emph{Adv.
  High Energy Phys.} {\bf 2014} (2014) 831598},
  [\href{http://arxiv.org/abs/1404.3751}{{\tt 1404.3751}}].

\bibitem{King:2017guk}
S.~F. King, \emph{{Unified Models of Neutrinos, Flavour and CP Violation}},
  \href{http://dx.doi.org/10.1016/j.ppnp.2017.01.003}{\emph{Prog. Part. Nucl.
  Phys.} {\bf 94} (2017) 217--256},
  [\href{http://arxiv.org/abs/1701.04413}{{\tt 1701.04413}}].

\bibitem{Cai:2017jrq}
Y.~Cai, J.~Herrero-Garc{\'\i}a, M.~A. Schmidt, A.~Vicente and R.~R. Volkas,
  \emph{{From the trees to the forest: a review of radiative neutrino mass
  models}}, \href{http://dx.doi.org/10.3389/fphy.2017.00063}{\emph{Front. in
  Phys.} {\bf 5} (2017) 63}, [\href{http://arxiv.org/abs/1706.08524}{{\tt
  1706.08524}}].

\bibitem{Casas:2001sr}
J.~A. Casas and A.~Ibarra, \emph{{Oscillating neutrinos and $\mu \to e,
  \gamma$}}, \href{http://dx.doi.org/10.1016/S0550-3213(01)00475-8}{\emph{Nucl.
  Phys. B} {\bf 618} (2001) 171--204},
  [\href{http://arxiv.org/abs/hep-ph/0103065}{{\tt hep-ph/0103065}}].

\bibitem{Cordero-Carrion:2018xre}
I.~Cordero-Carri{\'o}n, M.~Hirsch and A.~Vicente, \emph{{Master Majorana
  neutrino mass parametrization}},
  \href{http://dx.doi.org/10.1103/PhysRevD.99.075019}{\emph{Phys. Rev. D} {\bf
  99} (2019) 075019}, [\href{http://arxiv.org/abs/1812.03896}{{\tt
  1812.03896}}].

\bibitem{Cordero-Carrion:2019qtu}
I.~Cordero-Carri{\'o}n, M.~Hirsch and A.~Vicente, \emph{{General
  parametrization of Majorana neutrino mass models}},
  \href{http://dx.doi.org/10.1103/PhysRevD.101.075032}{\emph{Phys. Rev. D} {\bf
  101} (2020) 075032}, [\href{http://arxiv.org/abs/1912.08858}{{\tt
  1912.08858}}].

\bibitem{Heeck:2012fw}
J.~Heeck, \emph{{Seesaw parametrization for n right-handed neutrinos}},
  \href{http://dx.doi.org/10.1103/PhysRevD.86.093023}{\emph{Phys. Rev. D} {\bf
  86} (2012) 093023}, [\href{http://arxiv.org/abs/1207.5521}{{\tt 1207.5521}}].

\bibitem{Chen:2025cor}
Z.-Q. Chen, X.-H. Hu and Y.-L. Zhou, \emph{An exact parametrization of minimal
  seesaw model},  \href{http://arxiv.org/abs/2505.04279}{{\tt 2505.04279}},
  arXiv:2505.04279.

\bibitem{Minkowski:1977sc}
P.~Minkowski, \emph{{$\mu \to e\gamma$ at a Rate of One Out of $10^{9}$ Muon
  Decays?}}, \href{http://dx.doi.org/10.1016/0370-2693(77)90435-X}{\emph{Phys.
  Lett. B} {\bf 67} (1977) 421--428}.

\bibitem{Yanagida:1979as}
T.~Yanagida, \emph{{Horizontal gauge symmetry and masses of neutrinos}},
  {\emph{Conf. Proc. C} {\bf 7902131} (1979) 95--99}.

\bibitem{Mohapatra:1979ia}
R.~N. Mohapatra and G.~Senjanovic, \emph{{Neutrino Mass and Spontaneous Parity
  Nonconservation}},
  \href{http://dx.doi.org/10.1103/PhysRevLett.44.912}{\emph{Phys. Rev. Lett.}
  {\bf 44} (1980) 912}.

\bibitem{GellMann:1980vs}
M.~Gell-Mann, P.~Ramond and R.~Slansky, \emph{{Complex Spinors and Unified
  Theories}},  vol.~C790927, pp.~315--321, 1979.
\newblock \href{http://arxiv.org/abs/1306.4669}{{\tt 1306.4669}}.

\bibitem{Schechter:1980gr}
J.~Schechter and J.~W.~F. Valle, \emph{{Neutrino Masses in SU(2) x U(1)
  Theories}}, \href{http://dx.doi.org/10.1103/PhysRevD.22.2227}{\emph{Phys.
  Rev. D} {\bf 22} (1980) 2227}.

\bibitem{Ibarra:2003up}
A.~Ibarra and G.~G. Ross, \emph{{Neutrino phenomenology: The Case of two
  right-handed neutrinos}},
  \href{http://dx.doi.org/10.1016/j.physletb.2004.04.037}{\emph{Phys. Lett. B}
  {\bf 591} (2004) 285--296}, [\href{http://arxiv.org/abs/hep-ph/0312138}{{\tt
  hep-ph/0312138}}].

\bibitem{Atre:2009rg}
A.~Atre, T.~Han, S.~Pascoli and B.~Zhang, \emph{{The Search for Heavy Majorana
  Neutrinos}},
  \href{http://dx.doi.org/10.1088/1126-6708/2009/05/030}{\emph{JHEP} {\bf 05}
  (2009) 030}, [\href{http://arxiv.org/abs/0901.3589}{{\tt 0901.3589}}].

\bibitem{Lavoura:2003xp}
L.~Lavoura, \emph{{General formulae for $f_1 \to f_2 \gamma$}},
  \href{http://dx.doi.org/10.1140/epjc/s2003-01212-7}{\emph{Eur.Phys.J.C} {\bf
  29} (2003) 191--195}, [\href{http://arxiv.org/abs/hep-ph/0302221}{{\tt
  hep-ph/0302221}}].

\bibitem{Mohapatra:1986bd}
R.~N. Mohapatra and J.~W.~F. Valle, \emph{{Neutrino Mass and Baryon Number
  Nonconservation in Superstring Models}},
  \href{http://dx.doi.org/10.1103/PhysRevD.34.1642}{\emph{Phys. Rev. D} {\bf
  34} (1986) 1642}.

\bibitem{Gonzalez-Garcia:1988okv}
M.~C. Gonzalez-Garcia and J.~W.~F. Valle, \emph{{Fast Decaying Neutrinos and
  Observable Flavor Violation in a New Class of Majoron Models}},
  \href{http://dx.doi.org/10.1016/0370-2693(89)91131-3}{\emph{Phys. Lett. B}
  {\bf 216} (1989) 360--366}.

\bibitem{Akhmedov:1995ip}
E.~K. Akhmedov et~al., \emph{{Left-right symmetry breaking in NJL approach}},
  \href{http://dx.doi.org/10.1016/0370-2693(95)01504-3}{\emph{Phys.Lett.B} {\bf
  368} (1996) 270--280}, [\href{http://arxiv.org/abs/hep-ph/9507275}{{\tt
  hep-ph/9507275}}].

\bibitem{Akhmedov:1995vm}
E.~K. Akhmedov et~al., \emph{{Dynamical left-right symmetry breaking}},
  \href{http://dx.doi.org/10.1103/PhysRevD.53.2752}{\emph{Phys.Rev.D} {\bf 53}
  (1996) 2752--2780}, [\href{http://arxiv.org/abs/hep-ph/9509255}{{\tt
  hep-ph/9509255}}].

\bibitem{Malinsky:2005bi}
M.~Malinsky, J.~C. Romao and J.~W.~F. Valle, \emph{{Novel supersymmetric SO(10)
  seesaw mechanism}},
  \href{http://dx.doi.org/10.1103/PhysRevLett.95.161801}{\emph{Phys.Rev.Lett.}
  {\bf 95} (2005) 161801}, [\href{http://arxiv.org/abs/hep-ph/0506296}{{\tt
  hep-ph/0506296}}].

\bibitem{CentellesChulia:2024uzv}
S.~Centelles~Chuli{\'a}, A.~Herrero-Brocal and A.~Vicente, \emph{{The Type-I
  Seesaw family}}, \href{http://dx.doi.org/10.1007/JHEP07(2024)060}{\emph{JHEP}
  {\bf 07} (2024) 060}, [\href{http://arxiv.org/abs/2404.15415}{{\tt
  2404.15415}}].

\bibitem{Herrero-Brocal:2023czw}
A.~Herrero-Brocal and A.~Vicente, \emph{{The majoron coupling to charged
  leptons}}, \href{http://dx.doi.org/10.1007/JHEP01(2024)078}{\emph{JHEP} {\bf
  01} (2024) 078}, [\href{http://arxiv.org/abs/2311.10145}{{\tt 2311.10145}}].

\bibitem{Antusch:2023jok}
S.~Antusch, I.~Dor{\v{s}}ner, K.~Hinze and S.~Saad, \emph{{Fully testable axion
  dark matter within a minimal SU(5) GUT}},
  \href{http://dx.doi.org/10.1103/PhysRevD.108.015025}{\emph{Phys. Rev. D} {\bf
  108} (2023) 015025}, [\href{http://arxiv.org/abs/2301.00809}{{\tt
  2301.00809}}].

\end{thebibliography}\endgroup

\end{document}